\documentclass[11pt]{article}

\usepackage{amsthm}
\usepackage{amsfonts}
\usepackage{amsmath}
\usepackage{amssymb}
\usepackage{color}
\usepackage{easybmat}
\usepackage{blkarray}
\usepackage{datetime}

\usepackage{hyperref}
\usepackage{comment}
\usepackage{graphicx}
\usepackage[table]{xcolor}

\usepackage[T2A]{fontenc}
\usepackage[utf8]{inputenc}

\usepackage{caption}
\usepackage{floatrow}   
\usepackage{subcaption}

\usepackage{tikz}
\tikzstyle{boldd} = [line width=1.5]

\usepackage[colorinlistoftodos,prependcaption]{todonotes}
\presetkeys%
{todonotes}%
{inline}{}

\usepackage{mathrsfs}
\renewcommand{\mathcal}[1]{\mathscr{#1}}

\usepackage{marvosym}

\usepackage{theory}
\usepackage{theorems}

\renewcommand{\eps}{\varepsilon}

\newcommand{\E}{\operatorname{E}}

\newcommand{\err}{\mathsf{err}}

\title{
Towards Simpler Sorting Networks\\ and Monotone Circuits for Majority
}

\author{Natalia Dobrokhotova-Maikova$^1$ \and Alexander Kozachinskiy$^{2}$ \and Vladimir Podolskii$^3$}

\ddmmyyyydate

\date{%
 $^1$Yandex\\
    $^2$IMFD \& CENIA, Chile\\
    $^3$Tufts University
}
\begin{document}

\maketitle

\begin{abstract}
	In this paper, we study the problem of computing the majority function by low-depth monotone circuits and a related problem of constructing low-depth sorting networks. We consider both the classical setting with elementary operations of arity $2$ and the generalized setting with operations of arity $k$, where $k$ is a parameter. For both problems and both settings, there are various constructions known, the minimal known depth being logarithmic. However, there is currently no known construction that simultaneously achieves sub-log-squared depth, effective constructability, simplicity, and has a potential to be used in practice. In this paper we make progress towards resolution of this problem.
	
	For computing majority by standard monotone circuits (gates of arity 2) we provide an explicit monotone circuit of depth $O(\log_2^{5/3} n)$. The construction is a combination of several known and not too complicated ideas. 
	
	For arbitrary arity of gates $k$ we provide a new sorting network architecture inspired by representation of inputs as a high-dimensional cube. As a result we provide a simple construction that improves previous upper bound of $4 \log_k^2 n$ to $2 \log_k^2 n$. We prove the similar bound for the depth of the circuit computing majority of $n$ bits consisting of gates computing majority of $k$ bits. Note, that for both problems there is an explicit construction of depth $O(\log_k n)$ known, but the construction is complicated and the constant hidden in $O$-notation is huge.
\end{abstract}

\section{Introduction}

A sorting network receives an array of numbers and outputs the same numbers in the non-decreasing order.
It consists of \emph{comparators}, each of which is given some fixed pair of array entries as an input and it swaps them if they are not in the non-decreasing order. The main parameters of a sorting network are the size, that is, the number of comparators, and the depth, that is, the number of layers in the network, where each layer consists of several comparators applied to disjoint pairs of variables. Sorting networks are a classical model in theoretical computer science with vast literature devoted to them, see, for example~\cite{Batcher68,AjtaiKS83,Leighton85,
	paterson1990improved,Parberry92,KahaleLMPSS95,seiferas2009sorting,BundalaZ14}, see also the Knuth's book~\cite{Knuth98} and the Baddar's and Batcher's book~\cite{BaddarB12}. Despite considerable efforts, still there are many open problems related to sorting networks.
In this paper, our main interest is the depth of  sorting networks. 

There is a related setting of computing 
majority function by monotone Boolean circuits.  Majority function receives as input a sequence of $n$ bits and outputs 1 if and only if more than a half of the inputs are 1's. Monotone Boolean circuits consist of AND and OR gates of fan-in 2. 
Constructing a monotone Boolean circuit for the majority function can only be easier than constructing a sorting network. This is because a sorting network can be transformed into a monotone Boolean circuit which computes majority and has the same depth.
Indeed, if we restrict inputs to $\{0,1\}^n$, then each comparator can be simulated by a pair of AND and OR gates (AND computes the minimum of two Boolean inputs and OR computes the maximum). And the majority is just the median bit of the sorted array.

For the depth of sorting networks, there are several simple and practical constructions of depth $\Theta(\log^2 n)$~\cite{Knuth98,Batcher68,Parberry92}.
A construction with $O(\log n)$ depth was given by Ajtai, Koml{\'o}s and Szemer{\'e}dy~\cite{AjtaiKS83} and is usually referred to as the AKS sorting network. Although their bound on the depth is asymptotically optimal, the construction is very complicated and impractical due to a large constant hidden in the O-notation.
There are some simplifications and improvements of this construction~\cite{paterson1990improved,seiferas2009sorting}, but the construction is still elaborate and is not practical.
As for the lower bounds, there is a folklore
$(2 - o(1))\log_2 n$  depth lower bound for networks sorting $n$ numbers. 
It was improved by Yao~\cite{Yao80} and later by Kahale et al.~\cite{KahaleLMPSS95} with the current record about $3.27 \log_2 n$. 

As we discussed above, any construction for a sorting network translates to a monotone circuit for majority of the same depth. In particular, we get an  $O(\log n)$-depth monotone circuit for majority from the AKS sorting network. Yet again, the resulting circuit has the same disadvantages as the AKS construction. But in contrast to sorting networks, there is an alternative construction of a monotone depth-$O(\log n)$ Boolean circuit for majority due to Valiant~\cite{Valiant84}. His construction is simple and has a reasonable constant hidden in the O-notation, but it is randomized. It was partially derandomized and made closer to practice by Hoory, Magen and Pitassi~\cite{hoory2006monotone}. But still all known fully deterministic constructions that are simple and practical  are of depth $\Theta(\log^2 n)$.
Thus, there is an open problem for both sorting networks and monotone circuits for majority to come up with simple and deterministic construction of sub-log-squared depth.

One potential approach to this is to consider sorting networks with comparators that have $k > 2$ inputs. We will call them $k$-sorting networks.
They appear in the literature since 70s, the setting is mentioned already in the Knuth's book~\cite[Problem 5.3.4.54]{Knuth98}, followed by numerous works~\cite{TsengL85,ParkerP89,BeigelG90,NakataniHAT89,CypherS92,LeeB95,ShiYW14,gao1997sloping,zhao1998efficient}.
They are usually studied to better understand the structure of ordinary sorting networks (for example, a version of AKS sorting network with improved constant relies on $k$-sorting network in intermediate constructions~\cite{Chvatal92}).
In particular, $k$-sorting networks are closely related to recursive constructions of sorting networks. Having a good construction of a $k$-sorting network, one can apply it to its own comparators, getting a construction with smaller $k$, until eventually $k$ becomes 2, and we get an ordinary sorting network.

Parker and Parbery~\cite{ParkerP89} constructed a simple and potentially practical $k$-sorting network of depth $\leq 4 \log^2_k n$ (in case when $n$ is an integral power of $k$).
At the same time, as  Chv{\'a}tal shows in his lecture notes~\cite{Chvatal92}, 
the AKS sorting network also generalizes to this setting, giving a construction of depth $O(\log_k n)$. However, as with the AKS sorting network itself, this construction is complicated and impractical. 
So the search for simple constructions continues.
As for the lower bounds, any $k$-sorting network with $n$ inputs must have depth at least $\log_k n$ because otherwise outputs cannot be connected to all $n$ inputs. Dobokhotova-Maikova et al.~\cite{Dobrokhotova-MaikovaKP23} improved this bound to roughly $2 \log_k n$. They also found optimal values of $k$ for small values of depth $d$. More specifically, for sorting networks of depth $d=1,2$ they show that $k$ cannot be smaller than $n$, for $d=3$ the optimal value is $k = \left\lceil \frac n2 \right\rceil$ and for $d=4$ the optimal value is $k = \Theta(n^{2/3})$. These results indicate that small depth $k$-sorting networks are not enough for iterative approach to sub-log-squared sorting network and we need either good $k$-sorting network constructions of depth greater than 4, or additional ideas. 

Just as with sorting networks, we can consider circuits for majority function that are constructed from \emph{threshold} gates of fan-it at most $k$. A threshold gate is a Boolean function that first sorts its input bits in the non-decreasing order, and then outputs the $i$th one from the beginning, for some fixed $1\le i \le k$. For $k = 2$, AND and OR are the only two threshold functions. In general, there are $k$ threshold functions of fan-in $k$. By taking one copy of each, we get a comparator of arity $k$.
Thus, as in the case $k = 2$, a $k$-sorting network can be transformed into a circuit of the same depth which computes majority and consists of threshold gates of fan-in $k$. In other words, constructing a $k$-sorting network can only be harder than constructing a circuit for majority with threshold gates of fan-in $k$.

There is a line of work, initiated by Kulikov and Podolskii~\cite{KulikovP19}, which addresses the following question: given $d$ and $n$, what is the minimal $k$ for which there exists a circuit with threshold gates of fan-in $k$, which has depth $d$ and computes majority on $n$ bits? The paper~\cite{KulikovP19} shows that, up to a polylogarithmic factor, $k \ge n^{14/(7d + 6)}$. In subsequent works, special attention was given to the case $d = 2$. In this case, the lower bound of~\cite{KulikovP19} is $k\ge n^{14/20}$. It was improved to $k \ge n^{4/5}$ by Engels et al.~\cite{EngelsGMR20}. Then a linear lower bound $k \ge n/2 - o(n)$ was obtained by Hrubes et al.~\cite{HrubesRRY19}. An upper bound $k\le 2n/3 + O(1)$ was given in~\cite{Posobin17}. Let us also mention an upper bound $k\le n - 2$ for circuits that only use majority gates~\cite{kombarov2017,amano2017}.

Now, for $d \ge 3$ the situation is less clear. For $d = 3$, the paper~\cite{KulikovP19} gave an upper bound $k = O(n^{2/3})$. In turn, their lower bound in this case is of order $n^{14/27}$. We are not aware of any non-trivial upper bound for $d\ge 4$.

\paragraph*{Our results.} In this paper we make progress towards better constructions of monotone circuits for majority and sorting networks.

First, we give an explicit and reasonably simple construction of a monotone circuit for majority of depth $O(\log^{5/3} n)$. 
\begin{theorem} \label{thm:main1}
	There is a polynomial time constructable monotone circuit for majority of polynomial size and depth $O(\log^{5/3} n)$.
\end{theorem}
Our proof combines several relatively simple steps. We start with partial derandomization of Valiant's construction using some ideas from the paper by Cohen et al.~\cite{CohenDIKMRR13}. Next we apply to the resulting randomized circuit two operations several times. The first of them is a brute-force derandomization, that searches through all possible random bits of the randomized circuit. The second one is a composition with a $k$-sorting network of depth $O(\log^2_k n)$. For such a network we can use either the construction of Parker and Parbery~\cite{ParkerP89}, or, for the better constant, our next result.

In our second result we come up with the new architecture for $k$-sorting networks. As an application of this architecture we construct a $k$-sorting network of depth $2 \log_k^2 n$, improving the constant compared to the result of~\cite{ParkerP89}. More precisely, we prove the following theorem.

\begin{theorem} \label{thm:sorting_upper_bound}
	For any $n$ and for any $k$ such that $\log k = \omega(\log \log n)$ (or, to put it differently, $k$ is growing faster than any $\polylog(n)$), there exists a $k$-sorting network of depth at most $(2+o(1))\log^2_k n$.
\end{theorem}

The key idea behind this construction is to represent the input array as a hypercube of high dimension and sort various sections of this cube. We note that the idea of representing an array as a multidimensional structure is not new, for example, Leighton~\cite{Leighton85} in his ColumnSort represented the array as a two dimensional table. However, in our construction it is important that we use the dimension greater than $2$, since we use the fact that the sections of the cube have non-trivial intersection. On the conceptual level, the main novelty in our construction is the notion of $s$-sorting. We call the array $s$-sorted if the whole array is sorted correctly apart from some interval of length at most $s$. Most (if not all) log-squared-depth sorting network constructions adopt the divide and conquer strategy. The $O(\log_k^2 n)$-depth construction in~\cite{ParkerP89} is not an exception, to sort an array of size $n$, they split it into subarrays of size $n/k$, sort them recursively and merge them afterward. However, merging $k$ subarrays using $k$-sorting network is relatively expensive. To improve over previous construction, we work with $s$-sorted subarrays instead. We show how to merge them effectively (using the hypercube idea) and then show how we can build a recursive construction based on them.

To additionally illustrate applications of our construction, we consider constant depth sorting networks and circuits for majority. We show that there is a $\MAJ_k$-circuit for $\MAJ_n$ for $k = O(n^{3/5})$. For a second application we address the question of $k$-sorting networks for $k = O(n^{1/2})$. In~\cite{Knuth98} Knuth posed a problem of constructing a minimal depth $k$-sorting network for the input of size $k^2$. Parker and Parbery~\cite{ParkerP89} gave a construction of depth $9$. We improve this to depth $8$ at the cost of using comparators of size $O(k)$ for $k^2$ input size.

The rest of the paper is organized as follows. In Section~\ref{sec:prelim} we provide necessary preliminary information. In Section~\ref{sec:maj-circuit} we construct a monotone circuit for majority of depth $O(\log^{5/3}n)$. In Section~\ref{sec:k-sorting-network} we provide a new construction of $k$-sorting networks and deduce the corollaries. In Section~\ref{sec:conclusion} we discuss some open problems.


\section{Preliminaries}
\label{sec:prelim}

We use the standard notation $[n] = \{1,\ldots, n\}$. We sometimes omit the base of the logarithms, by default we assume that the base is $2$.

\subsection{Sorting Networks}

A depth-$d$ $k$-sorting network with $n$ inputs consists of $d + 1$ arrays $A_1, \ldots, A_{d+1}$, each of length $n$. Between any two arrays $A_i$  and $A_{i+1}$ there is a \emph{layer of comparators} (the first layer is between $A_1$ and $A_2$, the second layer is between $A_2$ and $A_3$, and so on). A layer of comparators is a partition of the set $\{1, 2, \ldots, n\}$ into subsets of size at most $k$ called \emph{comparators}. 

The input is given in an array $A_1$ and all other arrays are computed by the network one by one in the following way.
If $S \subseteq [n]$ is a comparator from the $i$th layer, then it  is applied to the entries $\{A_i[j] \mid j \in S\}$. It sorts their values in the non-decreasing order and puts the results into the entries $\{A_{i+1}[j] \in A_{i+1} \mid i \in S\}$. We say that a network is \emph{sorting} if for any input $A_1$ the array $A_{d+1}$ is sorted. 

We reserve the name \emph{sorting network} for $2$-sorting networks.

It is well known that to check that the sorting network sorts all possible inputs, it is enough to check that it sorts just $0/1$-inputs.
\begin{lemma}[Zero-one principle~\cite{Knuth98}]
	A network with $n$ inputs sorts all integer sequences in the non-decreasing order if and only if it sorts all sequences from $\{0,1\}^n$ in the non-decreasing order.
\end{lemma}

By this principle, when constructing sorting networks, we can assume that each input cell receives either $0$ or $1$.

The following simple observation will be useful to us.
\begin{lemma} \label{lem:sorting-monotonicity}
	If $t$ largest or $t$ smallest entries in the array are positioned correctly (i.e., in the last $t$ cells and in the first $t$ cells, respectively), then after the application of several comparators they are still positioned correctly.
\end{lemma}

\begin{proof}
    We can show by induction on $i$ that the smallest and the largest entries do not move if they are already positioned correctly. The key observation is that if some of these entries are inputted into one of the comparators $S$, they will not be moved.
\end{proof}

\subsection{From Sorting Networks to Majority Circuits}

We use the standard notion of Boolean circuits (see, e.g.~\cite{Jukna12}). As inputs, we allow Boolean variables and Boolean constants 0 and 1. The size of the circuit is the number of gates in it.

Given a $k$-sorting network we can get a circuit computing majority from it. More specifically, restrict the inputs to the network to $\{0,1\}^n$ and consider one $k$-comparator $S$. Note that its $k$th output is equal to 1 if and only if there is at least one 1 in the input. In other words, the $k$th output is equal to $\OR$ of input bits. Its $(k-1)$th output is equal to 1 if and only if there are at least two 1s in the input. More generally, it is easy to see that the $(k-i)$th output of the $k$-comparator outputs a threshold function
$$
\THR_k^{i}(x) = \begin{cases}
    1 & \text{if } |x| > i  ,\\
    0 & \text{otherwise,}
\end{cases}
$$
where $|x|$ denotes the weight of the vector $x \in \{0,1\}^k$, that is, the number of 1s in it. 
We reserve the notation $\MAJ_k(x)$ for the function $\THR_k^{k/2}(x)$.

We can substitute each comparator in the network by $k$ majority functions. Note that by adding several constants 0 or 1 as inputs to the gate we can convert any $\THR_k^{i}$ function into $\MAJ_{k'}$ with $k'\leq 2k$.

Now, it remains to observe that the median bit in the output array computes exactly $\MAJ_n$. 
Thus, as a result, we get the following lemma.
\begin{lemma} \label{lem:sorting-to-maj}
    Any $k$-sorting network of depth $d$ and size $s$ can be effectively converted into a circuit of depth $d$ and size $ks$ consisting of $\MAJ_{2k}$ gates and computing majority. In the case $k=2$, we get just a monotone circuit consisting of $\AND$ and $\OR$.
\end{lemma}

\subsection{Approximate Majority}

By $\eps$-approximate majority function $\MAJ^{\eps}_n$ we denote the partial function that outputs $\MAJ_n$ of its input but is defined only on the inputs where the fraction of ones in it is bounded away by $\eps$ from $1/2$. 

We need the following known result. 
\begin{theorem}[\cite{viola2009approximate}] \label{thm:approx_maj}
	For any constant $\eps > 0$, one can compute $\MAJ_n^\eps$ explicitly by a monotone circuit of size $\poly(n)$ and depth $O(\log n)$.
\end{theorem}

\subsection{$t$-Wise Independent Hash Functions}

We need the notion of $t$-wise independent hash functions.

\begin{definition} For integers $N$ and $t$ such that $t \leq N$, a family of function $\mathcal{H} = \{h \colon [N] \to [N]\}$ is $t$-wise independent if for all distinct $x_1,\ldots, x_t \in [N]$ the random variables $h(x_1), \ldots, h(x_t)$ are independent and uniformly distributed in $[N]$, when $h \in \mathcal{H}$ is drawn uniformly. 
\end{definition}

\begin{theorem}[\cite{Vadhan12}]
\label{thm:gen_hash}
	For every integer $n$ and $t$ such that $t \leq 2^n$ there is a family of $k$-wise independent functions $\mathcal{H} = \{h \colon \{0,1\}^n \to \{0,1\}^n\}$ such that choosing a random function from $\mathcal{H}$ takes $nt$ random bits and evaluating a function from $\mathcal{H}$ takes time $\poly(n,t)$.
\end{theorem}

\begin{theorem}[\cite{BellareR94}] \label{thm:hash-functions}
	Let $X$ be the average of $N$ $t$-wise independent random variables $X_1, \ldots, X_N \in [0,1]$ for even $t$. Then for any $\eps>0$ we have
	$$
	\Pr\left[\left|X - \E[X]\right| \geq \eps \right] \leq 1.1 \left( \frac{t}{N\eps^2}  \right)^{t/2}.
	$$
\end{theorem}

\section{Sub-log-squared Circuit for Majority}
\label{sec:maj-circuit}

In this section, we provide a proof of Theorem~\ref{thm:main1}.

Our goal is to compute $\MAJ_n$ by an explicit circuit of polynomial size and $o(\log^2 n)$ depth. We assume for convenience that $n$ is odd (for even $n$ we can consider a circuit for $n+1$ and substitute one variable by a constant). We start with some inferior circuit and perform several operations that allow us to gradually improve the parameters. However, on our way, we need to consider randomized circuits as well, and apart from size and depth, we will also be interested in the number of random bits and the error probability. More specifically, a circuit is an $(s,d,r,\err)$-circuit for majority if its size is at most $2^{s}$, depth is at most $d$, we can construct a circuit using at most $r$ random bits and the error probability on each input is at most $2^{-\err}$. Here all parameters are functions in the number of inputs $n$ (we write $\err = \infty$ when the circuit is correct with probability 1). All circuits we are going to consider are effectively constructible: there is an algorithm that given the values of random bits constructs a circuit in polynomial time in the size of the circuit.

Given a circuit with some parameters, we will use two operations to obtain new circuits. We are introducing these operations in the next two lemmas. Their effect on the circuit is summarized in the table below. 

\begin{center}
	\renewcommand{\arraystretch}{1.8}
	\begin{tabular}{|c|c|c|}\hline
		Initial circuit & Brute-force derandomization & Downward self-reduction\\ \hline
		$s(n)$ & $O(s(n) + r(n))$ &  $O(\log n) + s(2k)$\\ \hline
		$d(n)$ & $d(n) + O(r(n))$ &  $O\left(\left( \frac{\log_2 n}{\log_2 k} \right)^2 d(2k)\right)$\\ \hline
		$r(n)$ & $0$ &  $r(2k)$\\ \hline
		$\err(n)$ & $\infty$ &  $\err(2k) - O(\log n)$\\ \hline
	\end{tabular}
	\renewcommand{\arraystretch}{1}
\end{center}

\begin{lemma}[Brute-force derandomization] \label{lem:naive_derand}
	If there is an $(s,d,r,2)$-circuit $C$, then there is an $(O(s+r),d+O(r),0,\infty)$-circuit.
\end{lemma}

This lemma allows us to get rid of randomness but increases the depth and the size of the circuit if $r$ is large.

\begin{proof}
	Consider a randomized circuit $C_y(x)$, where $x\in \{0,1\}^n$ is an input and $y \in \{0,1\}^r$ is the sequence of random bits. Assume $C_y(x)$ has the parameters, as in the statement of the lemma.
	Consider circuits $C_y(x)$ for all possible values of $y$ and observe that for any $x$ the fraction of circuits that output $\MAJ_n(x)$ is at least $1 - 1/4 = 3/4$. Thus, if we feed $C_y(x)$ for all $y$ into a circuit from Theorem~\ref{thm:approx_maj} computing $\MAJ_{2^r}^\eps$, the output is exactly $\MAJ_n(x)$. 
	
	The size of the resulting circuit is at most $2^r \cdot 2^s + \poly(2^r)$, where the first term corresponds to computing $C_y(x)$ for all $y$ and the second term corresponds to computing $\MAJ_{2^r}^{\eps}$. Thus, the size is $2^{O(s+r)}$. Since all $C_y(x)$ can be computed in parallel, the depth of the circuit is at most $d + O(r)$.  The resulting circuit does not use random bits and is always correct.
\end{proof}

\begin{lemma}[Downward self-reduction]
	If there is   an $(s(n),d(n),r(n),\err(n))$-circuit $C$, then for any $k < n$ there is an $(O(\log n) + s(2k), O(\log_k^2 n \cdot d(2k)), r(2k), \err(2k) - O(\log n))$-circuit.
\end{lemma}

This operation increases the depth (if $d(n)$ is sub-log-squared), but allows to reduce other parameters.

\begin{proof}
	Consider a $k$-sorting network of depth $O(\log_k^2 n)$, given by~\cite{ParkerP89} or by our Theorem~\ref{thm:sorting_upper_bound} (the latter allows only for limited values of $k$, but the values we will actually use in the construction below are within the limits). 
	By Lemma~\ref{lem:sorting-to-maj} this network gives us a monotone circuit with the same parameters consisting of $\MAJ_{2k}$ gates computing $\MAJ_n$, denote this circuit by $C(x)$, where $x \in \{0,1\}^n$.

	Consider a $(s(2k), d(2k), r(2k), \err(2k))$-circuit $C_y$ on $k$ inputs, where $y \in \{0,1\}^{r(2k)}$. Fix $y$ and substitute each $\MAJ_{2k}$ gate in $C$ by $C_y$. Denote the resulting circuit by $D_y(x)$. This is a standard monotone Boolean circuit, its size is  $\mathrm{poly}(n) \cdot 2^{s(2k)}$, its depth is $O(\log^2_k n \cdot d(2k))$ and the number of random bits is $r(2k)$.
	
	It remains to show that the error probability is not too large. For this fix some input $x \in \{0,1\}^n$. Consider all $\MAJ_{2k}$ gates in $C(x)$ and denote their inputs when $x$ is fed to $C$ by $z^1, z^2, \ldots, z^t$. Here $t$ is the size of $C$ and is polynomial in $n$.
	
	For each $z^i$ the probability over random $y$ that $C_y(z^i)$ computes $\MAJ_{2k}(z^i)$ incorrectly is at most $2^{-\err(2k)}$. By union bound, with probability at least $1 - t 2^{-\err(2k)}$ we have $C_y(z^i) = \MAJ_k(z^i)$ for all $i$ and thus $D_y(x)$ computes $\MAJ_n(x)$ correctly. Thus, the probability of error of the resulting circuit is at most
	$$
	t \cdot 2^{-\err(2k)} = 2^{-\err(2k) + O(\log n)} 
	$$
\end{proof}

Now we describe our starting circuit. Interestingly, it is constructed as a partial derandomization of Valiant's construction.

\begin{lemma} \label{lem:start}
	There is an explicit circuit for majority with parameters $(O(\log n), O(\log n), O(\log^3 n), \Omega(\log^2 n))$.
\end{lemma}

We provide the proof of Lemma~\ref{lem:start} in Section~\ref{sec:proof-of-start} below, but before that, we explain how to finish the construction of the desired circuit for $\MAJ_n$.

Starting with the circuit provided by Lemma~\ref{lem:start}, we first apply downward self-reduction with the parameter $k$ satisfying $\log k = C \sqrt{\log n}$ for some big enough constant $C > 0$, then we apply brute-force derandomization, and then we apply downward self-reduction again with $k$ satisfying $\log k = \log^{2/3} n$.
We summarize the changes in the parameters after each step in the table below.
\begin{center}
	\renewcommand{\arraystretch}{1.8}
	\begin{tabular}{|p{0.83cm}||p{2.0cm}|p{2.4cm}|p{2.6cm}|p{2.5cm}|} \hline
		& Initial \mbox{circuit} & Step 1 & Step 2 & Step 3\\ \hline
		& & Self-reduction with \mbox{$\log k = \sqrt{\log n}$} & Brute-force \mbox{derandomization} & Self-reduction with \mbox{$\log k = \log^{2/3} n$}\\ \hline
		$s(n)$& $O(\log n)$ & $O(\log n)$ & $O(\log^{3/2} n)$ & $O(\log n)$ \\ \hline
		$d(n)$ & $O(\log n)$ & $O(\log^{3/2} n)$ & $O(\log^{3/2} n)$ & $O(\log^{5/3} n)$  \\ \hline
		$r(n)$& $O(\log^3 n)$  & $O(\log^{3/2} n)$ & $0$ & $0$ \\ \hline
		$\err(n)$& $\Omega(\log^2 n)$ & $\Omega(\log n)$ & $\infty$ & $\infty$ \\ \hline
	\end{tabular}
	\renewcommand{\arraystretch}{1}
\end{center}

\begin{remark}
	Note that with the two operations in hand, there are not that many options to apply them to a given initial construction. It is not hard to check that applying downward self-reduction two times in a row is not better than applying it once with the appropriate value of $k$. Clearly, there is no need to apply the derandomization step twice. From this, it is not hard to see that our sequence of operations is actually optimal. Once the optimal sequence of operations is established, it is not hard to check that our choice of parameters in downward self-reductions is optimal as well.
	
\end{remark}

\subsection{Proof of Lemma~\ref{lem:start}}
\label{sec:proof-of-start}

In this subsection, we are going to prove Lemma~\ref{lem:start}. The high-level idea is to partially derandomize Valiant's construction. To make the presentation self-contained we first recall the idea behind this construction.

Suppose we have independent random bits $x,y,z$ that are equal to 1 with probability $p$ and consider $\MAJ_3(x,y,z)$. It is not hard to see that it outputs 1 with probability $f(p) = p^3 + 3p^2(1-p)$. Consider $p = \frac 12 + \eps$ for some $\eps>0$ and denote $\eps' = f(p) - \frac 12$. Then
$$
\eps' = f(p) - \frac 12 = f(p) - f(\frac 12) = f'(\alpha)(p - \frac 12) = f'(\alpha) 
\eps
$$
for some $\alpha \in [\frac 12, p]$.
Note that $f'(p) = 6p - 6p^2 = 6p(1-p)$. It is easy to see that for $\alpha \in [\frac 12, \frac 23]$ we have $f'(\alpha) \geq \frac 43$.
Thus, for $\eps \in [\frac 12, \frac 23]$ we have $\eps' \geq \frac 43 \eps$.

Now, we can use this in the following way. Consider 
$\MAJ_n$ for odd $n$ and consider its arbitrary input $x$. Without loss of generality, assume that $\MAJ_n(x)=1$. If we draw one variable from $x$ uniformly at random, it is equal to 1 with probability at least $\frac 12 + \frac 1n$. Consider a $\MAJ_3$ gate and feed to it three independently and uniformly drawn input variables. By the analysis above the output of such a $\MAJ_3$ gate is equal to 1 with probability at least $\frac 12 + \frac 43 \cdot \frac 1n$. Now we can repeat this: consider three such $\MAJ_3$ gates and feed their outputs to another $\MAJ_3$ gates. The result is equal to 1 with probability $\frac 12 + \left( \frac 43 \right)^2 \frac 1n$. After $O(\log n)$ many iterations, we get a $O(\log n)$-depth randomized circuit consisting of $\MAJ_3$ gates that output the correct value with probability at least $\frac 23$. Valiant's argument further improves this probability, but we will not need this part of the argument.

The randomized circuit above uses too many random bits. Now we are going to modify the construction in a way, that uses randomness more efficiently. We will use some ideas from~\cite{CohenDIKMRR13}. 

Construct the following circuit consisting of $\MAJ_3$ gates. The circuit contains $\Theta(\log n)$ layers, each containing $N = n^3$ gates. The bottom layer consists of input variables,  each  repeated $\frac Nn = n^2$ times (it is redundant to copy variables several times, we do this exclusively for the sake of uniformity of the construction). In other layers, each gate computes the $\MAJ_3$ function of some gates from the previous layer. To assign the inputs to each gate, for each layer $j$ we draw three fresh (and independent of each other) $t$-wise independent hash functions $f_j,g_j,h_j \colon [N] \to [N]$ for $t = \Theta(\log n)$.
For a gate with number $i$ in layer $j$ we set its inputs to be gates with numbers $f_j(i)$, $g_{j}(i)$ and $h_{j}(i)$ in layer $(j-1)$.

Before we finish the construction of the circuit, let us analyze the current part.
Consider some input $x \in \{0,1\}^n$, assume without loss of generality that $\MAJ_n(x)=1$.
Denote by $\frac 12 + \eps_i$ the fraction of gates on level $i$ that output 1. For $i=1$ we have $\eps_1 \geq \frac 1n$.

Each gate on level $i$ receives three independent inputs from the previous level. Thus, the probability that it outputs 1 is at least $\frac 12 + \frac 43 \eps_{i-1}$ (we have shown this above only for $\eps_{i-1} \leq \frac 23$, but these values of $\eps_{i-1}$ are enough for our construction as well). Thus, the expected fraction of ones in level $i$ is also at least $\frac 12 + \frac 43 \eps_{i-1}$.

Now we would like to use concentration inequality to show that with high probability the fraction of correct values is not much smaller than its expectation. Note that the outputs of the gates on level $i$ are t-wise independent.

Let $\eps = \frac{1}{6n}$ and denote by $X_i$ the output of $i$-th gate. Then by Theorem~\ref{thm:hash-functions} we have
$$
\Pr\left[\left|\sum_i X_i/N - (\frac 12 + (4/3)\eps_{j-1})\right| \geq \eps \right] \leq 1.1\left( \frac{t}{N\eps^2}  \right)^{t/2} = 2^{- \Theta(\log^2 n)}.
$$

By union bound the probability that on each level $\eps_j \geq \frac 43 \eps_{j-1} - \eps$ is at least
$$
1- O(\log n) \cdot 2^{-\Theta(\log^2 n)} = 1 - 2^{-\Theta(\log^2 n)}.
$$
Thus, we can show by induction on $j$ that with probability at least $1 - 2^{\Theta(\log^2 n)}$ we have
$$
\eps_j \geq \frac 43 \eps_{j-1} - \eps \geq \frac 76 \eps_{j-1} + \frac 16 \eps_{j-1} - \frac{1}{6n} \geq \frac 76 \eps_{j-1},
$$
where in the last inequality we use that by induction hypothesis we have $\eps_{j-1} \geq \left(\frac 76\right)^{j-1}\cdot \eps_1 \geq \frac 1n$.

Thus, just like in Valiant's argument,  after $O(\log n)$ iterations, with porbability $1 - 2^{-\Theta(\log^2 n)}$, we have $\eps_j \geq \frac 23$. At this point, it remains to apply to the last layer a circuit from Theorem~\ref{thm:approx_maj}. 

It is easy to see that the size  of the resulting circuit is $\poly(n)$, the depth is $O(\log n)$, error probability is $2^{- \Theta(\log^2 n)}$. 
As for the random bits, note that in the construction we need $O(\log n)$ $t$-wise independent hash functions from $[N]$ to $[N]$. By Theorem~\ref{thm:gen_hash} there are families of such functions defined using $O(t \log N)$ random bits. In total we need
$$
O(\log n) \cdot O(t \log N) = O(\log^3 n)
$$
random bits.

This finishes the proof of Lemma~\ref{lem:start}.

\begin{remark}
    Instead of applying a circuit for Approximate Majority to the last layer, we could do the following: sample $m = O(\log^2 n)$ gates from the last layer uniformly at random and then compute the majority on these $m$ gates using some simple circuit of depth $O(\log^2 m)$. By Chernoff's inequality, this adds at most $2^{-\Omega(m)} = 2^{- \Theta(\log^2 n)}$ to the error probability, and we need $O(\log^3 n)$ random bits. In turn, the increase in depth and size is negligible.
\end{remark}

\section{$k$-Sorting Network Construction}
\label{sec:k-sorting-network}

\subsection{Proof Strategy}

Before we proceed to the proof we would like to illustrate the idea considering some specific value of $k$. For convenience, we assume that $n$ is a cube of a natural number.

\begin{lemma}
	Assume that $n=t^3$ for natural $t$. Then there is a depth-4 $k$-sorting network with $k = 2t^2 = 2 n^{2/3}$.
\end{lemma}

We present the proof using a geometric interpretation of an input array as a three-dimensional cube. However, note that a similar result is implicit in~\cite{Leighton85} and it is essentially the same construction, just in different terms. We also note that it is known that this is the optimal (up to a constant factor) value of $k$ for depth-4 sorting networks~\cite{Dobrokhotova-MaikovaKP23}.

\begin{figure}
	\centering
	\begin{subfigure}[t]{.3\linewidth}
		\centering
		\begin{tikzpicture}[scale=0.5]
			\foreach \x in {0,...,5}
			{   \draw (0,\x ,5) -- (5,\x ,5);
				\draw (\x ,0,5) -- (\x ,5,5);
				\draw (5,\x ,5) -- (5,\x ,0);
				\draw (\x ,5,5) -- (\x ,5,0);
				\draw (5,0,\x ) -- (5,5,\x );
				\draw (0,5,\x ) -- (5,5,\x );
			}
        \draw[->] (0,0,5) -- (0,0,7) node[above] {$y$};
        \draw[->] (0,5,0) -- (0,6,0) node[right] {$x$};
        \draw[->] (5,0,0) -- (6,0,0) node[above] {$z$};
		\end{tikzpicture}
		\caption{Input array}
		\label{fig:input}
	\end{subfigure}
	\hfill
	\begin{subfigure}[t]{.3\linewidth}
		\centering
		\begin{tikzpicture}[scale=0.5]
			\foreach \x in {0,...,5}
			{   \draw[color=black!20] (0,\x ,5) -- (5,\x ,5);
				\draw[color=black!20] (5,\x ,5) -- (5,\x ,0);
				\draw[color=black!20] (5,0,\x ) -- (5,5,\x );
				\draw[color=black!20] (0,5,\x ) -- (5,5,\x );
			}
			\foreach \x in {0,...,5}
			{	\draw[boldd] (\x ,0,5) -- (\x ,5,5) -- (\x ,5,0);
			}
			\draw[boldd] (0,0 ,5) -- (5,0 ,5) -- (5,0,0) -- (5,5,0) -- (0,5,0) -- (0,5,5) -- (0,0,5);
            \draw[->] (0,0,5) -- (0,0,7) node[above] {$y$};
            \draw[->] (0,5,0) -- (0,6,0) node[right] {$x$};
            \draw[->] (5,0,0) -- (6,0,0) node[above] {$z$};
		\end{tikzpicture}
		\caption{Step 1: cut the cube into vertical slices}\label{fig:step1}
	\end{subfigure}
	\hfill
	\begin{subfigure}[t]{.3\linewidth}
		\centering
		\begin{tikzpicture}[scale=0.5]
			\foreach \x in {0,...,5}
			{   \draw[color=black!20] (0,\x ,5) -- (5,\x ,5);
				\draw[color=black!20] (5,\x ,5) -- (5,\x ,0);
				\draw[color=black!20] (\x ,0,5) -- (\x ,5,5);
				\draw[color=black!20] (\x ,5,5) -- (\x ,5,0);
			}
			\foreach \x in {0,...,5}
			{
				\draw[boldd] (5,0,\x ) -- (5,5,\x ) -- (0,5,\x );
			}
			\draw[boldd] (0,0 ,5) -- (5,0 ,5) -- (5,0,0) -- (5,5,0) -- (0,5,0) -- (0,5,5) -- (0,0,5);
            \draw[->] (0,0,5) -- (0,0,7) node[above] {$y$};
            \draw[->] (0,5,0) -- (0,6,0) node[right] {$x$};
            \draw[->] (5,0,0) -- (6,0,0) node[above] {$z$};
		\end{tikzpicture}
		\caption{Step 2: cut the cube into vertical slices in the other direction}\label{fig:step2}
	\end{subfigure}
	
	\bigskip
	\begin{subfigure}{0.45\linewidth}
		\centering
		\begin{tikzpicture}[scale=0.5]
			\foreach \x in {0,...,5}
			{ 	\draw[color=black!20] (\x ,0,5) -- (\x ,5,5);
				\draw[color=black!20] (\x ,5,5) -- (\x ,5,0);
				\draw[color=black!20] (5,0,\x ) -- (5,5,\x );
				\draw[color=black!20] (0,5,\x ) -- (5,5,\x );
			}
			\foreach \x in {0,2,4,5}
			{ 	\draw[boldd] (0,\x ,5) -- (5,\x ,5) -- (5,\x ,0);
			}
			\draw[boldd] (0,0 ,5) -- (5,0 ,5) -- (5,0,0) -- (5,5,0) -- (0,5,0) -- (0,5,5) -- (0,0,5);
            \draw[->] (0,0,5) -- (0,0,7) node[above] {$y$};
            \draw[->] (0,5,0) -- (0,6,0) node[right] {$x$};
            \draw[->] (5,0,0) -- (6,0,0) node[above] {$z$};
		\end{tikzpicture}
		\caption{Step 3: cut the cube into horizontal slices of width 2}\label{fig:step3}
	\end{subfigure} 
	\hfill
	\begin{subfigure}{.45\linewidth}
		\centering
		\begin{tikzpicture}[scale=0.5]
			\foreach \x in {0,...,5}
			{ 	\draw[color=black!20] (\x ,0,5) -- (\x ,5,5);
				\draw[color=black!20] (\x ,5,5) -- (\x ,5,0);
				\draw[color=black!20] (5,0,\x ) -- (5,5,\x );
				\draw[color=black!20] (0,5,\x ) -- (5,5,\x );
			}
			\foreach \x in {0,1,3,5}
			{ 	\draw[boldd] (0,\x ,5) -- (5,\x ,5) -- (5,\x ,0);
			}
			\draw[boldd] (0,0 ,5) -- (5,0 ,5) -- (5,0,0) -- (5,5,0) -- (0,5,0) -- (0,5,5) -- (0,0,5);
            \draw[->] (0,0,5) -- (0,0,7) node[above] {$y$};
            \draw[->] (0,5,0) -- (0,6,0) node[right] {$x$};
            \draw[->] (5,0,0) -- (6,0,0) node[above] {$z$};
		\end{tikzpicture}
		\caption{Step 4: horizontal slices with a shift}\label{fig:step4}
	\end{subfigure}
	\RawCaption{\caption{Sorting network for $k = 2 n^{2/3}$ (here $n=125$, $k = 50$ and $t=5$)}
		\label{fig:3-dimensional}}
\end{figure}
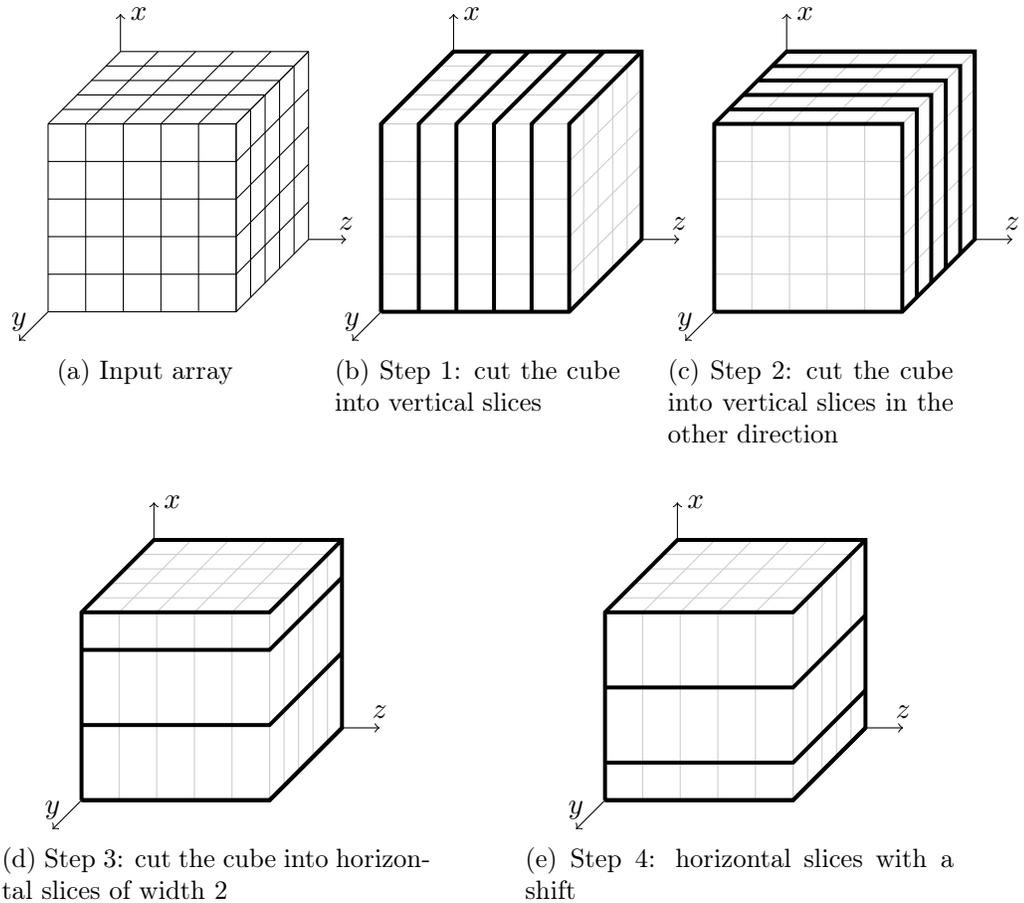

\begin{proof}
	We represent entries of an input array as a 3-dimensional cube with the side $t$ (see Figure~\ref{fig:input}). We place the first $t^2$ entries of an array in the bottom layer of the cube, the next $t^2$ entries in the second layer of the cube and so on. In each layer the entries are positioned row by row.

    To be more precise, assume that the array $A$ is enumerated as $[a_1, \ldots, a_n]$. We reenumerate the same array as
    $$
    [a_{111}, a_{112}, \ldots, a_{11t}, a_{121}, \ldots, a_{12t}, \ldots, a_{ttt}].
    $$
    That is, entries of an array are enumerated by sequences $(x,y,z) \in \{1,\ldots,t\}^3$ in the lexicographic order. In Figure~\ref{fig:3-dimensional} $a_{xyz}$ corresponds to a subcube with coordinates $(x,y,z)$.
	
	In the first layer of the sorting network we split the cube into vertical slices of width 1 and feed each slice to a $t^2$-comparator (see Figure~\ref{fig:step1}). To be more precise, for each $i=1,\ldots, t$ we feed entries $a_{xyi}$ for all $x,y$ into one comparator. On the second layer of the network we split the cube into vertical slices of width 1 in another direction and feed each slice to a $t^2$-comparator (see Figure~\ref{fig:step2}). In other works, for each $i=1,\ldots, t$ we feed entries $a_{xiz}$ for all $x,z$ into one comparator. On the third layer we split the cube into horizontal slices of width 2 (for odd $t$ the last slice is of width 1) and feed the slices to comparators of arity at most $2t^2$ (see Figure~\ref{fig:step3}). Finally, on the fourth layer of the network we split the cube into horizontal slices of width 2 again, but now the first slice is of width 1 (for even $t$ the last slice is of width 1 as well). Thus, the slices on this layer are shifted compared to the previous one (see Figure~\ref{fig:step4}).
	
	It remains to prove that this sorting network sorts correctly. Consider any input $x \in \{0,1\}^{n}$. Note that the cube consists of $t^2$ vertical columns with $t$ entries in each column: each column $A_{yz}$ is obtained by fixing $y$ and $z$ in $a_{xyz}$ and considering all possible $x$. We are interested in the weight $w_{yz}$ of each column, that is the number of 1s in it. For the input $A$ the weights of the columns can be any numbers from $0$ to $t$. Now consider the array after the first layer of the network. Note that now each vertical slice of the first layer of the network is sorted. This means that in each of these slices in the first several  rows (from bottom to top) there are only 0s, then there might be a row containing both 0s and 1s, and then all remaining rows contain 1s. In particular, the weights of two columns in the same slice differ by at most~1.
	
	Now consider the second layer of the network and consider two different slices $S_i = \{a_{x,i,z} \mid x,z \in [t]\}$ and $S_j = \{a_{x,j,z} \mid x,z \in [t]\}$. Note that each of them contains exactly one column from each slice of the first layer. We know that the weights of the columns in the same slice of the first layer differ by at most 1. Thus, in total, the number of 1s in two slices of the second layer differ by at most $t$. In other words, for each $z$ the first slice contains the column $A_{iz}$ and the second slice contains the column $A_{jz}$. We know that on the input of the second layer of the network $|w_{iz} - w_{jz}| \leq 1$. Thus, 
    $$
    |\sum_z w_{iz} - \sum_z w_{jz}| \leq t.
    $$
    Denote by $r_i$ the number of rows in slice $S_i$ that consists of only $1$s after the second layer of the network. We just showed that the slice $S_i$ can have one more extra row of 1s, one less row of 1s or something in between. Overall, for the number $r_j$ of rows consisting of $1$s in $S_j$ we have $|r_i - r_j| \leq 1$.
    As a result, the weights of columns in slices $S_i$ and $S_j$ can differ by at most 2. Since this is true for any $i$ and $j$, we have that the weights of all columns in the cube after the second layer of the sorting network differ by at most 2. 
    
    To put it another way, there is a horizontal slice of width 2, such that below this slice we have only 0s and above this slice we have only 1s. Thus it remains to sort entries of this slice. Note that on layers 3 and 4 of the network there is a comparator that sorts exactly this slice. Note that by Lemma~\ref{lem:sorting-monotonicity} all other comparators of layers 3 and 4 do not harm the sorting.
\end{proof}

This argument can be extended to the cubes of arbitrary dimension $d$. More specifically, for $n = t^d$ and for $k = (d-1)t^{d - 1}$ we can represent entries of an input array as a $d$-dimensional cube with side $d$, sort `vertical' slices (we need to fix one of the coordinates in $d$-dimensional space as vertical) in all $d-1$ directions and then sort horizontal slices. This results into $(d-1)$ layers of the sorting network and for horizontal slices we need recursive calls for the arrays of size approximately $2dt^{d-1}$. Actually, it is expensive to make two recursive calls for horizontal layers, instead we use an additional trick to make just one recursive call.

Although our $k$-sorting network construction can be expressed in terms of high dimensional hypercubes, we prefer to give a more general exposition, using a concept of \emph{$s$-sorted arrays}.  

\subsection{Merging $s$-Sorted Arrays}
The following definition plays a key role in our sorting network construction.
\begin{definition}
    A $0/1$-array $A$ of length $n$ is \emph{$s$-sorted} if there is an integer interval $I = \{i, \ldots, i+s-1\} \subseteq [n]$, such that $A[j]=0$ for $j<i$ and $A[j]=1$ for $j \geq i+s$. We call $I$ \emph{unsorted interval}.
\end{definition}

As an immediate corollary of Lemma~\ref{lem:sorting-monotonicity}, we get the following.
\begin{corollary} \label{cor:s-sorted-stability}
	Suppose a sorting network gets an $s$-sorted array with 
 unsorted interval $I$. Then the output is also $s$-sorted with $I$ as an unsorted interval.
\end{corollary}


We give a construction of a depth-1 sorting network that ``merges'' $p$ arrays of length $n$ that are already $s$-sorted into one array which is $(sp + O(np^2/k))$-sorted, where $k$ is the arity of the sorting network.

\begin{lemma} \label{lem:mergin-s-sorted}
	Assume that $k \geq tp$ for some integers $t$ and $p$. Suppose we have $p$ $s$-sorted arrays of size $n$ each. Assume additionally that $n$ is divisible by $t$. Then there is a depth-1 $k$-sorting network that merges these arrays into one array of size $np$ that is $(sp+2\frac{np}{t})$-sorted. If additionally we assume that $s$ is divisible by $n/t$, then the resulting array is $(sp+\frac{np}{t})$-sorted.
\end{lemma}

\begin{proof}
	Represent each array as a table with $\frac nt$ columns and $t$ rows. We assume the following ordering on the entries of this table: to compare two entries, we first compare the indices of their rows, and then the indices of their columns. Position the tables one under another in a unified table with $tp$ rows. Note that $tp \leq k$ and apply $k$-comparator to each column in parallel. We claim that the resulting array in the large table is $(sp+2\frac{np}{t})$-sorted.
	
	To see that observe, that in each small table, an unsorted interval of length at most $s$ occupies at most $\left\lceil \frac{st}{n} \right\rceil +1$ rows (any other row either consists entirely of 0s or entirely of 1s). In the large table, this gives us at most $p\left( \left\lceil \frac{st}{n} \right\rceil +1 \right)$ non-constant rows.  After sorting each column individually, 0-rows will move to the top, 1-rows will move to the bottom and all other $p\left( \left\lceil \frac{st}{n} \right\rceil +1 \right)$ rows will be in between on them. They constitute an unsorted interval and the size of it is at most
	$$
	\frac nt \cdot p \left( \left\lceil \frac{st}{n} \right\rceil +1 \right).
	$$
	
	For general $s$ we can upper bound this as follows
	$$
	\frac nt \cdot p \left( \left\lceil \frac{st}{n} \right\rceil +1 \right) \leq \frac nt \cdot p \left( \frac{st}{n} + 2\right) = sp + 2\frac{np}{t}.
	$$
	
	If $s$ is divisible by $n/t$, note that we can just drop the rounding operation and the size of an unsorted interval is at most
	$$
	\frac nt \cdot p \left( \left\lceil \frac{st}{n} \right\rceil +1 \right) = \frac nt \cdot p \left( \frac{st}{n} + 1\right) = sp + \frac{np}{t} .
	$$
\end{proof}

Applying previous lemma several times we get the following.

\begin{lemma} \label{lem:s-sorting}
	Consider arbitrary $n$ and $k$ and denote $t = \lfloor \sqrt{k} \rfloor$. Then there is a $k$-sorting network of depth $\lceil \log_t n \rceil - 1$ that on any input outputs an $s$-sorted array for $s \leq \frac{2 \left\lceil \log_{t} n \right\rceil n}{t}$.
\end{lemma}

\begin{proof}
	Denote $d = \lceil \log_t n \rceil$ and observe that $n \leq t^d$. Introduce the following notation:
	$$
	n_i = \begin{cases}
		t^{i+1} & \text{for } i=1,\ldots, d-2,\\
		t^{d-1} p & \text{for } i=d-1,
	\end{cases}
	$$
	where $p$ is such that $t^{d-1} (p-1) < n \leq t^{d-1} p$. In particular, since $p\geq 2$, we have $p-1 \geq p/2$ and
	$$
	n > t^{d-1} (p-1) \geq t^{d-1} p/2.
	$$
	For the convenience of presentation, we add $t^{d-1} p - n$ dummy inputs equal to 1 to the end of the array to make the size of the input to be equal to $t^{d-1} p$. By Lemma~\ref{lem:sorting-monotonicity} these inputs will never change their position and can be removed from the sorting network.
	
	We start with an unsorted array as an input and repeatedly apply Lemma~\ref{lem:mergin-s-sorted} to get the array consisting of blocks that are $s$-sorted for some $s$. More specifically, after level $i$ of the network we get the blocks of size $n_i$ that are $s_i$ sorted for
	$$
	s_i = \begin{cases}
		(i-1)t^{i} & \text{for } i=1,\ldots, d-2,\\
		(d-2)t^{d-2} p & \text{for } i=d-1.
	\end{cases}
	$$
	On the first step we split the input into blocks of size $t^2$ and apply comparators to them, the resulting blocks are $0$-sorted.
	
	On the $i$-th step for $i=2, \ldots, d-1$ we already have blocks of size $n_{i-1} = t^i$ from the previous step that are $s_{i-1}$-sorted for $s_{i-1} = (i-2)t^{i-1}$. Note that $n_{i-1} = t^i$ is divisible by $t$ and $s_{i-1}$ is divisible by $n_i/t = t^{i-1}$. We apply Lemma~\ref{lem:mergin-s-sorted} and for $i<d-1$ get blocks of size $n_{i-1} t = n_{i}$ that are $s$-sorted for $s = s_{i-1}t + n_{i-1} = (i-1)t^{i}$. For $i=d-1$ we have just $p$ subarrays to merge and after the step we get the whole array of size $t^{d-1}p$ that is $s$-sorted for $s = (d-3)t^{d-2} p + \frac{t^{d-1}p}{t} = (d-2)t^{d-2}p$.
	
	Finally, observe that
	$$
	s \leq (d-2)t^{d-2} p \leq d \frac{2n}{t}
	$$
	as desired.
\end{proof}

\subsection{Computing Majority}

Before constructing a sorting network we solve a simpler task of computing majority function.

\begin{theorem} \label{thm:maj_upper_bound}
For any $n$ and for any $k$ such that $\log k = \omega(\log \log n)$ (or, to put it differently, $k$ is growing faster than any $\polylog(n)$), there exists a $\MAJ_k$-circuit for $\MAJ_n$ of depth at most $(2+o(1))\log^2_k n$.
\end{theorem}

The rest of the section is devoted to the proof of Theorem~\ref{thm:maj_upper_bound}.

First observe that to compute $\MAJ_n$ correctly by a monotone circuit it is enough to compute it correctly on minterm and maxterm inputs: the computation on other inputs follows by monotonicity. Thus, we can assume in our construction that the input contains almost the same number of 0s and 1s. We will construct a sorting network that sorts all such inputs correctly. From the sorting network we get the circuit of the same depth.

Suppose we need to sort an array of size $n$ with approximately the same number of 0s and 1s. We apply Lemma~\ref{lem:s-sorting} to the array. This results in a $Y$-sorted array for 
$Y = \frac{2 \left\lceil \log_{t} n \right\rceil n}{t}$ for $t=\lfloor \sqrt{k} \rfloor$.
Since the number of 0s and 1s in the array is approximately equal, the smallest $\frac n2 - Y$ and the largest $\frac n2 -Y$ elements are sorted correctly (otherwise, the length of unsorted interval is larger than $Y$). Thus, it remains to sort specific interval of length $2Y$ and we can do it recursively. 

Overall, we get the following recursive relation.
$$
T(n) \leq \lceil \log_t n \rceil-1 + T\left(2Y\right) \leq 
\log_t n + T\left(\frac{4 \lceil \log_t n \rceil n}{t}\right).
$$

To solve this recursive relation we use the following lemma.

\begin{lemma}\label{lem:recur}
	Assume that $\log k = \omega(\log \log n)$. Suppose that $T(n) = \const$ for $n$ up to some constant and
	$$
	T(n) \leq 2 \log_{k} n  + C + T\left(\left\lceil\frac{D  (\log_k n) n}{\sqrt{k}}\right\rceil\right)
	$$
	for some constants $C$ and $D>0$. Then $T(n) \leq (2 + o(1))\log_k^2 n$.
\end{lemma}

\begin{proof}
	To simplify the presentation, we ignore rounding of the argument of $T$ first, and address it later. Denote $\alpha = \frac{\sqrt{k}}{D \log_k n}$.
	
	
	We have
	\begin{align*}
		T(n) &\leq 2\log_k n + C + T\left(\frac{n}{\alpha}\right)\\
		&\leq 2\log_k n + C + 2\log_k \frac{n}{\alpha} + C + T\left(\frac{n}{\alpha^2}\right)\\
		&\leq 2 \sum_{i=0}^{\log_{\alpha} n}  \left(\log_k \frac{n}{\alpha^i} + C \right)\\
		&= 2 \left( \log_k n + (\log_k n - \log_k \alpha) + (\log_k n - 2 \log_k \alpha) + \ldots + 0 \right) + 2 C \log_{\alpha} n\\
		&\le 2 \frac{\log_{k} n}{\log_k \alpha}  \frac{\log_k n}{2} + 2 C \log_{\alpha} n=
		\log_{k}^2 n \log_{\alpha} k + 2 C \log_{\alpha} n.
	\end{align*}

    It is easy to see that the term $2 C \log_{\alpha} n$ is negligible, since $\alpha \gg k^{1/3}$.
	
	We analyze $\log_{\alpha} k$ factor separately:
	\begin{align*}
		\log_{\alpha} k &= \log_{\frac{\sqrt{k}}{D  \log_k n }} k 
		= \frac{\log_2 k}{\log_2 \frac{\sqrt{k}}{D \log_k n }} 
		\leq \frac{\log_2 k}{\frac 12 \log_2 k - D - \log_2 \log_k n}\\
		&= \frac{\log_2 k}{\frac 12 \log_2 k - D - \log_2 \log_2 n + \log_2 \log_2 k}.
	\end{align*}
	
	For $\log k = \omega(\log \log n)$ this term is $2+o(1)$ and we have
	$$
	T(n) \leq  (2+o(1))\log_k^2 n.
	$$
	
	To address the rounding operation, note that $\left\lceil \frac{n}{\alpha}\right\rceil \leq \frac{n}{\alpha} + 1 \leq \frac{2n}{\alpha}$ for $\frac{n}{\alpha} \geq 1$. Thus, in the presence of rounding we will have $\sum_i \log_k \frac{2n}{\alpha}$ in the calculation above instead of $\sum_i \log_k \frac{n}{\alpha}$. This amounts to substituting $D$ by $2D$ and does not change the result of the calculation since $D$ is an arbitrary constant.
\end{proof}

\subsection{Constructing Sorting Network}
\label{sec:finish-sorting-network}

In this section, we finish the proof of Theorem~\ref{thm:sorting_upper_bound}. 

We adopt the same strategy as for the computation of majority. More specifically, we apply Lemma~\ref{lem:s-sorting} recursively to get $s$-sorted array for smaller and smaller $s$. However,  now our task is more tricky. In the proof of Theorem~\ref{thm:maj_upper_bound} when we get to an $s$-sorted array we know exactly where the unsorted interval is located (in the middle of the array). However, now we need to sort arbitrary input arrays and an unsorted interval can be anywhere.

We construct the network recursively. We assume that at the beginning of each step, we have an $s$-sorted array (at the beginning of the process $s=n$). Denote the unsorted interval by $A$, $|A| \leq s$.
Split the array into consecutive blocks $B_1,\ldots, B_p$ of size $s$ (the last block $B_p$) might be smaller. 
	
The recursive step consists of two stages. In the first stage, we split the array into blocks $B_1 \cup B_2$, $B_3, \cup B_4$, and so on, each block of size $2s$ (one last block might be smaller). In the second stage, we split the array into blocks $B_1$, $B_2 \cup B_3$, $B_4 \cup B_5$, and so on (again the last block might be smaller than $2s$). Before describing each of the stages, observe that either in the first stage or in the second stage (or in both) the interval $A$ falls completely into one of the blocks. Indeed, $A$ can intersect with at most two consecutive blocks $B_i$, $B_{i+1}$ and in one of the stages, they form a single block.

In the first stage, we apply Lemma~\ref{lem:s-sorting} to each of the blocks $B_1 \cup B_2,B_3, \cup B_4, \ldots$ separately. As a result, each block is $s'$-sorted for $s' \le \frac{4 \left\lceil \log_{\lfloor \sqrt{k} \rfloor} n \right\rceil s}{\lfloor \sqrt{k} \rfloor}$. Moreover, if the block consisted of only 0s and 1s, then it does not change.

If $A$ is contained in one of the blocks of the first stage, we are already done: there is only one initially unsorted block that by Lemma~\ref{lem:s-sorting} after the step is $s'$-sorted. By Corollary~\ref{cor:s-sorted-stability} this property remains true after the additional comparators we apply for the other case.
	
If $A$ is split between two blocks of the first stage, then after the stage we have two consecutive unsorted blocks, each of them is $s'$-sorted. Denote unsorted parts by $C_1, C_2$. Note that by Corollary~\ref{cor:s-sorted-stability} $C_1,C_2 \subseteq A$ and thus, $C_1$ and $C_2$ fall into one block of the second stage. It is tempting to apply Lemma~\ref{lem:s-sorting} to the blocks of the second stage as well. However, this application is too expensive and will not result in the desired bound.
	
Instead we do the following. We represent each block of the second stage (of size at most $2s$) as a table with $p = \left\lceil2s/k\right\rceil$ columns and $k$ rows, filled in row by row from top to bottom. For convenience, if the last row is not complete, add dummy variables equal to 1 to complete the row.
	
Each of the intervals $C_1, C_2$ occupy at most $\left\lceil s'/p\right\rceil + 1$ rows. There might be another row that contain a switch between blocks $B_i$ and $B_{i+1}$. Each other row consist either entirely of 0s, or entirely of 1s. Denote the number of all 0 rows by $a$ and the number of all 1 rows by $b$.
	
We apply a comparator to each column separately. As a result, each column will contain $a$ zeros in the beginning, $b$ ones in the end and some part in between. The number of rows in the middle part is at most $2 \left\lceil s'/p\right\rceil + 3$. The number of entries in these rows is at most
$$
s'' = p(2 \left\lceil s'/p\right\rceil + 3) \leq 2s' + 5p \leq 3s'
$$
for large enough input size.
Thus, after the second stage we get $s''$-sorted array and we are done with the recursion step.
	
Thus, we get that $s'' \leq 12 \frac{\left\lceil \log_{t} n \right\rceil s}{t}$ and we get the following recursive relation
$$
T(n) \leq 
\log_t n + T\left(\frac{12 \lceil \log_t n \rceil n}{t}\right).
$$
We apply Lemma~\ref{lem:recur} again to get $T(n) \leq (2 + o(1))\log_k^2 n$.	

This finishes the proof of Theorem~\ref{thm:sorting_upper_bound}.

\subsection{Other Applications}

In this section we give two more examples of results that follow from our construction.

\begin{lemma} \label{lem:depth-4-majority}
    There is a $\MAJ_k$-circuit of depth $4$ computing $\MAJ_n$ for $k = O(n^{3/5})$.
\end{lemma}

\begin{proof}
    Denote $r = \lceil n^{1/5} \rceil$. For simplicity we pad the input with constants 0 and 1 to make the size of the array $r^5$ without changing the output of majority. We will use $k$-sorters for $k = 4r^3$.

    As in the proof of Theorem~\ref{thm:maj_upper_bound} it is enough to compute $\MAJ_n$ on minterms and maxterms, thus we can assume that there are approximately equal number of 0s and 1s in the input. 

    We will build a $k$-sorting network and the existence of the circuit follows. On the first layer of the network we split the input into blocks of size $r^3$ and sort them. On the second layer we use Lemma~\ref{lem:mergin-s-sorted} with $p=r$ and $t=r^2$. As a result we get blocks of size $r^4$ that are $r^2$-sorted. On the third level we apply Lemma~\ref{lem:mergin-s-sorted} again with the same values of $p$ and $t$. As a result we have that the whole input is now $2 r^3$-sorted. On the last layer of the network just as in the proof of Theorem~\ref{thm:maj_upper_bound} we apply $4r^3$-comparator to the middle of the array.
\end{proof}

In~\cite{Knuth98} Knuth posed a problem of constructing a minimal depth $k$-sorting network for the input of size $k^2$. Parker and Parbery~\cite{ParkerP89} gave a construction of depth $9$. Here we slightly improve on this at the cost of using comparators of size $O(k)$.
\begin{lemma}
    There is a $k$-sorting network of depth $8$ that sorts an array of size $n$ with $k = O(n^{1/2})$.
\end{lemma}

\begin{proof}
    As usual pad an array with constants to make $n=r^4$ for some integer $r$. Thus $k = O(r^2)$. 

    We follow the same strategy as in Section~\ref{sec:finish-sorting-network}.
    First we apply Lemma~\ref{lem:s-sorting} that uses three layers of network and results in a $s$-sorted array for $s = O(r^3)$. Then, we apply Lemma~\ref{lem:s-sorting} again to the blocks of size $O(r^3)$ to get a network of depth 2 that results in each block being $O(r^2)$-sorted. Then we apply one more layer to merge unsorted intervals in different blocks to get the array that is $O(r^2)$-sorted. Finally, we again split the array into blocks, this time of size $O(r^2)$ to complete the sorting using two layers. In total we use $3+2+1+2=8$ layers.
\end{proof}

\section{Conclusion}
\label{sec:conclusion}

The obvious open problems are to come up with explicit constructions of sorting networks and monotone circuits for majority of smaller depth. One specific problem, is to extend our $O(\log^{5/3} n)$ construction to sorting networks. The obstacle that we encountered is that there is no randomized construction of low depth sorting network that we can use as a start. Another interesting question is to extend our $O(\log^{5/3}n)$ construction to get a $\MAJ_k$-circuit for $\MAJ_n$ of depth $O(\log^{5/3}_k n)$. Such a construction can be used instead of $O(\log^2_k n)$-depth circuit in downward self-reduction to further improve the upper bound. Again the obvious obstacle is that it is not clear how to get a starting construction.


\begin{thebibliography}{10}

\bibitem{AjtaiKS83}
Mikl{\'{o}}s Ajtai, J{\'{a}}nos Koml{\'{o}}s, and Endre Szemer{\'{e}}di.
\newblock Sorting in c log n parallel steps.
\newblock {\em Comb.}, 3(1):1--19, 1983.

\bibitem{amano2017}
Kazuyuki Amano and Masafumi Yoshida.
\newblock Depth two $(n-2)$-majority circuits for $n$-majority.
\newblock Preprint, 2017.

\bibitem{BaddarB12}
S.W.A.H. Baddar and K.E. Batcher.
\newblock {\em Designing Sorting Networks: A New Paradigm}.
\newblock SpringerLink : B{\"u}cher. Springer New York, 2012.

\bibitem{Batcher68}
Kenneth~E. Batcher.
\newblock Sorting networks and their applications.
\newblock In {\em American Federation of Information Processing Societies:
  {AFIPS} Conference Proceedings: 1968 Spring Joint Computer Conference,
  Atlantic City, NJ, USA, 30 April - 2 May 1968}, volume~32 of {\em {AFIPS}
  Conference Proceedings}, pages 307--314. Thomson Book Company, Washington
  {D.C.}, 1968.

\bibitem{BeigelG90}
Richard Beigel and John Gill.
\newblock Sorting n objects with a k-sorter.
\newblock {\em {IEEE} Trans. Computers}, 39(5):714--716, 1990.

\bibitem{BellareR94}
Mihir Bellare and John Rompel.
\newblock Randomness-efficient oblivious sampling.
\newblock In {\em 35th Annual Symposium on Foundations of Computer Science,
  Santa Fe, New Mexico, USA, 20-22 November 1994}, pages 276--287. {IEEE}
  Computer Society, 1994.

\bibitem{BundalaZ14}
Daniel Bundala and Jakub Zavodny.
\newblock Optimal sorting networks.
\newblock In Adrian{-}Horia Dediu, Carlos Mart{\'{\i}}n{-}Vide, Jos{\'{e}}~Luis
  Sierra{-}Rodr{\'{\i}}guez, and Bianca Truthe, editors, {\em Language and
  Automata Theory and Applications - 8th International Conference, {LATA} 2014,
  Madrid, Spain, March 10-14, 2014. Proceedings}, volume 8370 of {\em Lecture
  Notes in Computer Science}, pages 236--247. Springer, 2014.

\bibitem{Chvatal92}
V.~Chv{\'a}tal.
\newblock Lecture notes on the new aks sorting network.
\newblock Technical report, 1992.

\bibitem{CohenDIKMRR13}
Gil Cohen, Ivan~Bjerre Damg{\aa}rd, Yuval Ishai, Jonas K{\"{o}}lker, Peter~Bro
  Miltersen, Ran Raz, and Ron~D. Rothblum.
\newblock Efficient multiparty protocols via log-depth threshold formulae -
  (extended abstract).
\newblock In Ran Canetti and Juan~A. Garay, editors, {\em Advances in
  Cryptology - {CRYPTO} 2013 - 33rd Annual Cryptology Conference, Santa
  Barbara, CA, USA, August 18-22, 2013. Proceedings, Part {II}}, volume 8043 of
  {\em Lecture Notes in Computer Science}, pages 185--202. Springer, 2013.

\bibitem{CypherS92}
Robert Cypher and Jorge L.~C. Sanz.
\newblock Cubesort: {A} parallel algorithm for sorting {N} data items with
  s-sorters.
\newblock {\em J. Algorithms}, 13(2):211--234, 1992.

\bibitem{Dobrokhotova-MaikovaKP23}
Natalia Dobrokhotova{-}Maikova, Alexander Kozachinskiy, and Vladimir~V.
  Podolskii.
\newblock Constant-depth sorting networks.
\newblock In Yael~Tauman Kalai, editor, {\em 14th Innovations in Theoretical
  Computer Science Conference, {ITCS} 2023, January 10-13, 2023, MIT,
  Cambridge, Massachusetts, {USA}}, volume 251 of {\em LIPIcs}, pages
  43:1--43:19. Schloss Dagstuhl - Leibniz-Zentrum f{\"{u}}r Informatik, 2023.

\bibitem{EngelsGMR20}
Christian Engels, Mohit Garg, Kazuhisa Makino, and Anup Rao.
\newblock On expressing majority as a majority of majorities.
\newblock {\em {SIAM} J. Discret. Math.}, 34(1):730--741, 2020.

\bibitem{gao1997sloping}
Qingshi Gao and Zhiyong Liu.
\newblock Sloping-and-shaking.
\newblock {\em Science in China Series E: Technological Sciences},
  40(3):225--234, 1997.

\bibitem{hoory2006monotone}
Shlomo Hoory, Avner Magen, and Toniann Pitassi.
\newblock Monotone circuits for the majority function.
\newblock In {\em Proceedings of the 9th international conference on
  Approximation Algorithms for Combinatorial Optimization Problems, and 10th
  international conference on Randomization and Computation}, pages 410--425,
  2006.

\bibitem{HrubesRRY19}
Pavel Hrubes, Sivaramakrishnan~Natarajan Ramamoorthy, Anup Rao, and Amir
  Yehudayoff.
\newblock Lower bounds on balancing sets and depth-2 threshold circuits.
\newblock In Christel Baier, Ioannis Chatzigiannakis, Paola Flocchini, and
  Stefano Leonardi, editors, {\em 46th International Colloquium on Automata,
  Languages, and Programming, {ICALP} 2019, July 9-12, 2019, Patras, Greece},
  volume 132 of {\em LIPIcs}, pages 72:1--72:14. Schloss Dagstuhl -
  Leibniz-Zentrum f{\"{u}}r Informatik, 2019.

\bibitem{Jukna12}
Stasys Jukna.
\newblock {\em Boolean Function Complexity - Advances and Frontiers}, volume~27
  of {\em Algorithms and combinatorics}.
\newblock Springer, 2012.

\bibitem{KahaleLMPSS95}
Nabil Kahal{\'{e}}, Frank~Thomson Leighton, Yuan Ma, C.~Greg Plaxton, Torsten
  Suel, and Endre Szemer{\'{e}}di.
\newblock Lower bounds for sorting networks.
\newblock In Frank~Thomson Leighton and Allan Borodin, editors, {\em
  Proceedings of the Twenty-Seventh Annual {ACM} Symposium on Theory of
  Computing, 29 May-1 June 1995, Las Vegas, Nevada, {USA}}, pages 437--446.
  {ACM}, 1995.

\bibitem{Knuth98}
Donald~Ervin Knuth.
\newblock {\em The art of computer programming, , Volume III, 2nd Edition}.
\newblock Addison-Wesley, 1998.

\bibitem{kombarov2017}
Yu.~A. Kombarov.
\newblock On depth two circuits for the majority function.
\newblock In {\em Proceedings of Problems in theoretical cybernetics}, pages
  129--132. Max Press, 2017.

\bibitem{KulikovP19}
Alexander~S. Kulikov and Vladimir~V. Podolskii.
\newblock Computing majority by constant depth majority circuits with low
  fan-in gates.
\newblock {\em Theory Comput. Syst.}, 63(5):956--986, 2019.

\bibitem{LeeB95}
De{-}Lei Lee and Kenneth~E. Batcher.
\newblock A multiway merge sorting network.
\newblock {\em {IEEE} Trans. Parallel Distributed Syst.}, 6(2):211--215, 1995.

\bibitem{Leighton85}
Frank~Thomson Leighton.
\newblock Tight bounds on the complexity of parallel sorting.
\newblock {\em {IEEE} Trans. Computers}, 34(4):344--354, 1985.

\bibitem{NakataniHAT89}
Toshio Nakatani, Shing{-}Tsaan Huang, Bruce~W. Arden, and Satish~K. Tripathi.
\newblock K-way bitonic sort.
\newblock {\em {IEEE} Trans. Computers}, 38(2):283--288, 1989.

\bibitem{Parberry92}
Ian Parberry.
\newblock The pairwise sorting network.
\newblock {\em Parallel Process. Lett.}, 2:205--211, 1992.

\bibitem{ParkerP89}
Bruce Parker and Ian Parberry.
\newblock Constructing sorting networks from k-sorters.
\newblock {\em Inf. Process. Lett.}, 33(3):157--162, 1989.

\bibitem{paterson1990improved}
Michael~S Paterson.
\newblock Improved sorting networks with o (logn) depth.
\newblock {\em Algorithmica}, 5(1):75--92, 1990.

\bibitem{Posobin17}
Gleb Posobin.
\newblock Computing majority with low-fan-in majority queries.
\newblock {\em CoRR}, abs/1711.10176, 2017.

\bibitem{seiferas2009sorting}
Joel Seiferas.
\newblock Sorting networks of logarithmic depth, further simplified.
\newblock {\em Algorithmica}, 53(3):374--384, 2009.

\bibitem{ShiYW14}
Feng Shi, Zhiyuan Yan, and Meghanad~D. Wagh.
\newblock An enhanced multiway sorting network based on n-sorters.
\newblock In {\em 2014 {IEEE} Global Conference on Signal and Information
  Processing, GlobalSIP 2014, Atlanta, GA, USA, December 3-5, 2014}, pages
  60--64. {IEEE}, 2014.

\bibitem{TsengL85}
S.~S. Tseng and Richard C.~T. Lee.
\newblock A parallel sorting scheme whose basic operation sorts\emph{N}
  elements.
\newblock {\em Int. J. Parallel Program.}, 14(6):455--467, 1985.

\bibitem{Vadhan12}
Salil~P. Vadhan.
\newblock Pseudorandomness.
\newblock {\em Found. Trends Theor. Comput. Sci.}, 7(1-3):1--336, 2012.

\bibitem{Valiant84}
Leslie~G. Valiant.
\newblock Short monotone formulae for the majority function.
\newblock {\em J. Algorithms}, 5(3):363--366, 1984.

\bibitem{viola2009approximate}
Emanuele Viola.
\newblock On approximate majority and probabilistic time.
\newblock {\em Computational Complexity}, 18:337--375, 2009.

\bibitem{Yao80}
Andrew~Chi{-}Chih Yao.
\newblock Bounds on selection networks.
\newblock {\em {SIAM} J. Comput.}, 9(3):566--582, 1980.

\bibitem{zhao1998efficient}
Lijun Zhao, Zhiyong Liu, and Qingshi Gao.
\newblock An efficient multiway merging algorithm.
\newblock {\em Science in China Series E: Technological Sciences},
  41(5):543--551, 1998.

\end{thebibliography}
\end{document}